\documentclass[preprint]{elsarticle}

\usepackage{array,xspace,multirow,hhline,tikz,colortbl,tabularx,booktabs,fixltx2e,amsmath,amssymb,amsfonts,amsthm}
\usetikzlibrary{arrows}
\usepackage{algorithm}
\usepackage{algorithmic}
\usepackage{eqparbox}
\usepackage{verbatim,ifthen}
\usepackage{enumitem}
\usepackage{pifont}
\usepackage{ifthen}
\usepackage{calrsfs,mathrsfs}
\usepackage{bbding,pifont}
\usepackage{pgflibraryshapes}

	\usepackage{varioref}
		\usepackage{subfigure}

\definecolor{light-gray}{gray}{0.9}

\newcommand{\w}{e\xspace}

	\newtheorem{lemma}{Lemma}%
	\newtheorem{theorem}{Theorem}%
	\newtheorem{proposition}{Proposition}%
	\newtheorem{corollary}{Corollary}%
	\newtheorem{example}{Example}
    	
	\newcommand\eat[1]{}

	\usepackage{enumitem}
	\setenumerate[1]{label=\rm(\it{\roman{*}}\rm),ref=({\it\roman{*}}),leftmargin=*}
	\newlength{\wordlength}
	\newcommand{\wordbox}[3][c]{\settowidth{\wordlength}{#3}\makebox[\wordlength][#1]{#2}}

	\newcommand{\midd}{\mathbin{:}}

	\newcommand{\eqclass}[2][]{\ifthenelse{\equal{#1}{}}{[#2]}{[#2]_{\sim_{#1}}}}

	\newcommand{\pref}{\succsim\xspace}

\usepackage{enumitem}
\setenumerate[1]{label=\rm(\it{\roman{*}}\rm),ref=({\it\roman{*}}),leftmargin=*}

\newcommand{\nbh}[1][]{
	\ifthenelse{\equal{#1}{}}{\nu}{\nu(#1)}
}

\newcommand{\cstr}[1][]{
	\ifthenelse{\equal{#1}{}}{\mathscr S}{\cstr(#1)}
}

\newcommand{\choice}[1][]{
	\ifthenelse{\equal{#1}{}}{\mathit{C}}{\choice(#1)}

		\newcommand{\ml}[1][]{\ensuremath{\ifthenelse{\equal{#1}{}}{\mathit{ML}}{\mathit{ML}(#1)}}\xspace}
		\newcommand{\sml}[1][]{\ensuremath{\ifthenelse{\equal{#1}{}}{\mathit{SML}}{\mathit{SML}(#1)}}\xspace}
		\newcommand{\sd}[1][]{\ensuremath{\ifthenelse{\equal{#1}{}}{\mathit{SD}}{\mathit{SD}(#1)}}\xspace}
		\newcommand{\rsd}[1][]{\ensuremath{\ifthenelse{\equal{#1}{}}{\mathit{RSD}}{\mathit{RSD}(#1)}}\xspace}
		\newcommand{\rd}[1][]{\ensuremath{\ifthenelse{\equal{#1}{}}{\mathit{RD}}{\mathit{RD}(#1)}}\xspace}
		\newcommand{\st}[1][]{\ensuremath{\ifthenelse{\equal{#1}{}}{\mathit{ST}}{\mathit{ST}(#1)}}\xspace}
		\newcommand{\bd}[1][]{\ensuremath{\ifthenelse{\equal{#1}{}}{\mathit{BD}}{\mathit{BD}(#1)}}\xspace}
		\newcommand{\pc}[1][]{\ensuremath{\ifthenelse{\equal{#1}{}}{\mathit{PC}}{\mathit{PC}(#1)}}\xspace}
		\newcommand{\dl}[1][]{\ensuremath{\ifthenelse{\equal{#1}{}}{\mathit{DL}}{\mathit{DL}(#1)}}\xspace}
		\newcommand{\ul}[1][]{\ensuremath{\ifthenelse{\equal{#1}{}}{\mathit{UL}}{\mathit{UL}(#1)}}\xspace}

			\newcommand{\indiff}{\ensuremath{\sim}}}

\begin{document}

		\title{Mechanisms for House Allocation\\ with Existing Tenants under Dichotomous Preferences}

		% \title{Generalizing Top Trading Cycles\\ for Housing Markets with Fractional Endowments\\\vspace{1em} \small Draft: Please do not cite or circulate. Comments welcome!}

%: Computational and Manipulation Aspects  }%\tnoteref{t1}} 
	%\tnotetext[t1]{Thanks!} 
	% \tnotetext[t2]{The second title footnote which is a longer than the first one and with an intention to fill
	% in up more than one line while formatting.}

	\author{Haris Aziz\corref{cor1}} \ead{haris.aziz@nicta.com.au}
	% \author{Markus Brill} \ead{brill@in.tum.de}
	% 	\author{Paul Harrenstein} \ead{harrenst@in.tum.de}
		\address{Data61 and UNSW, Sydney, Australia}

%	\fntext[fn1]{Thanks!} 
%%

% TODO:
% \begin{itemize}
% \item prove that \FTTC is strategyproof for dichotomous preferences.
% \item prove that \FTTC allocation is core stable for lexicographic preferences.
% %\item prove that \FTTC allocation with uniform endowments is equivalent to the egalitarian mechanism of \citet{BoMo04a}.
% \end{itemize}

	\begin{keyword}
	 House allocation\sep 
	Housing markets\sep
	Core\sep
	Top Trading Cycles\sep 
	Dichotomous preferences\\
		
		\emph{JEL}: C62, C63, and C78
	\end{keyword}

    \begin{abstract}
        We consider house allocation with existing tenants in which each agent has dichotomous preferences. %We first point out that previously presented algorithms for the problem do not satisfy maximum number of agents or are not guaranteed to be strategyproof. 
        We present strategyproof, polynomial-time, and (strongly) individually rational algorithms that satisfy the maximum number of agents. For the endowment only model, one of the algorithms also returns a core-stable allocation. 
            \end{abstract}

	\maketitle

		%\today

		\sloppy

%\today

    \section{Introduction}
    
    %Principled allocation of indivisible goods is a central problem in multi-agent settings~\citep{BCM15a, SoUn10a}.
    We consider a setting in which there is a set of agents $N=\{1,\ldots, n\}$, set of houses $H=\{h_1,\ldots, h_m\}$, each agent owns at most one and possibly zero houses from $H$ and there may be houses that are not owned by any agent.  In the literature the setting is termed as \emph{house allocation with existing tenants}~\citep{AbSo99a,SoUn10b}.  If each agent owns exactly one house and each house is owned, the setting is equivalent to the \emph{housing market setting}~\citep{ACMM05a,Ma94a,ShSc74a}. If no agent owns a house, the setting is equivalent to the house allocation model~\citep{HyZe79a, Sven94a}.
    
    The setting captures various allocations settings including allocation of dormitory rooms as well as kidney exchange markets. 
We consider housing market with existing tenants setting with the assumption that each agent has dichotomous preferences with utility zero or one. 

Dichotomous preferences are prevalent in many settings. A kidney may either be compatible or incompatible for a patient. A dormitory room may within budget or outside budget for a student. \citet{BoMo04a} give numerous other examples where dichotomous preferences make sense in allocation and matching problems. Interestingly, since dichotomous preferences require indifference in preferences, they are not covered by many models considered in the literature that assume strict preferences~\citep{AbSo99a,SoUn10b}. 
If a house is maximally preferred, we will refer to it as an acceptable house.
If an agent gets a house that is maximally preferred, we will say that the agent is \emph{satisfied}. 

We consider the standard normative properties in market design: 
(i) \emph{Pareto optimality}: there should be no allocation in which each agent is at least as happy and at least one agent is strictly happier (ii) \emph{individual rationality (IR)}: no agent should have an incentive to leave the allocation program (iii) \emph{strategyproofness}: no agent should have an incentive to misreport his preferences. We also consider a stronger notion of individual rationality that requires that either an agent keeps his endowment or gets some house that is strictly better.
Finally, we consider a property stronger than Pareto optimality: maximizing the number of agents who are satisfied. If we rephrase dichotomous preferences in terms of 1-0 utilities, the goal is equivalent to maximizing social welfare. Under 1-0 utilities, an allocation maximizes welfare if a maximum number of agents get utility 1 or in other words a maximum number of agents are satisfied.  Throughout the paper, when we talk about welfare, we will assume that the underlying cardinal utilities are 1-0. 

The main results are as follows:

 \begin{theorem}
     For house allocation with existing tenants with dichotomous preferences, 
     % \begin{enumerate}
     %              \item 
there exists an algorithm (MSIR) that is polynomial-time, strategyproof, and allocates maximally preferred houses to a maximum number of agents subject to strong individual rationality. In the endowment setting in which each house is initially owned by an agent, the mechanisms are core stable.  
     %     \item there exists an algorithm (MIR) that is polynomial-time, strategyproof, maximizes welfare subject to individual rationality, and is Pareto optimal.
     %
     % \end{enumerate}
    \end{theorem}
     
     \begin{theorem}
         For house allocation with existing tenants with with dichotomous preferences, there exists an algorithm (MIR) that is polynomial-time, strategyproof, allocates maximally preferred houses to a maximum number of agents, and is Pareto optimal. 
         \end{theorem}

    \section{Related Work}

    House allocation with existing tenants model was introduced by \citet{AbSo99a}. They assumed that agents have strict preferences and introduced mechanism that are individually rational, Pareto optimal and strategyproof. The results do not apply to the setting with dichotomous preferences 
    
In the restricted model of housing market, \citet{JaMa12a} proposed a mechanism called TCR that is polynomial-time, individually rational, Pareto optimal and strategyproof even if agents have indifferences. The results imply the same result for dichotomous preferences. %However, we show that TCR does not satisfy the maximum number of agents that is a central concern such as in kidney exchange. 

Another natural approach to allocation problems is to satisfy the maximum number of agents by computing a maximum weight matching. However, such an approach violates individual rationality. 
\citet{KMRZ+14a} proposed algorithms for house allocation with 1-0 utility. However their algorithms are for the model in which agents do not have any endowment.

For kidney exchange with 1-0 utilities, \citet{ABS07a} and \citet{BMR09a} presented an algorithm that is strongly individually rational and satisfies the maximum number of agents. The algorithm does not cater for the cases where agents own acceptable houses, or there may be houses that are not owned by any agents. Furthermore, they did not establish strategyproofness or core stability of their mechanism. 
    
    % \subsection{Previous Approaches}
   %
   %  \begin{itemize}
   %      \item TTC works for strict preferences.
   %      \item TCR: does not satisfy the maximum number of agents subject to individual rationality.
   %      \item Maximum cardinaality matching may not be inidividually rational. Also it may not be strategyproof.
   %      \item Algorithm by \citet{ABS07a} for kidney exchange works is individually rational and satisfies the maximum number of agents. However it was not shown to be strategyproof.
   %  \end{itemize}
    
    \section{Preliminaries}

    Let $N=\{1,\ldots, n\}$ be a set of $n$ agents and $H=\{h_1,\ldots, h_m\}$ a set of $m$ houses. The endowment function $\omega: N\rightarrow H$ assigns to each agent the house he originally owns. Each agent has complete and transitive preferences $\pref_i$ over the houses and $\pref=(\pref_1,\ldots \pref_n)$ is the preference profile of the agents. The \emph{housing market} is a quadruple $M=(N, H, \w, \pref)$. For $S\subseteq N$, we denote $\bigcup_{i\in S}\{\w(i)\}$ by $\w(S)$.
    A function $x:S\rightarrow H$ is an \emph{allocation} on $S\subseteq N$ if there exists a function $\pi$ on $S$ such that $x(i)=\w(\pi(i))$ for each $i\in S$. The goal in housing markets is to re-allocate the houses in a mutually beneficial and efficient way.

    Since we consider dichotomous preferences, we will denote by $A_i$ the set of houses maximally preferred by $i\in N$. These houses are \emph{acceptable} to $i$ whereas the other houses are unacceptable to $i$. Agent $i$ is \emph{satisfied} when he gets an acceptable house. Since we consider welfare as well, we will assume that that acceptable houses give utility one to the agent and unacceptable houses give utility zero. Hence getting an unacceptable house is equivalent to not getting a house.

    %Next, we present a number of solution concepts for house allocation. 
    %\paragraph{Solution concepts:}
    An allocation $x$ is \emph{individually rational (IR)} if $x(i)\pref_i \w(i)$ for all $i\in N$
     An allocation $x$ is \emph{strongly individually rational (S-IR)} if $x(i)= \w(i)$ or $x(i) \succ_i (\w(i))$ for all $i\in N$.
    
    A coalition $S\subseteq N$ \emph{blocks} an allocation $x$ on $N$ if there exists an allocation $y$ on $S$ such that for all $i\in S$, $y(i)\in \w(S)$ and $y(i)\succ_i x(i)$. An allocation $x$ on $N$ is in the \emph{core (C)} of market $M$ if it admits no blocking coalition. An allocation that is in the core is also said to be \emph{core stable}.
    % An allocation is \emph{weakly Pareto optimal (w-PO)} if $N$ is not a blocking coalition.
     A coalition $S\subseteq N$ \emph{weakly blocks} an allocation $x$ on $N$ if there exists an allocation $y$ on $S$ such that for all $i\in S$, $y(i)\in \w(S)$, $y(i)\pref_i x(i)$, and there exists an $i\in S$ such that $y(i)\succ_i x(i)$. An allocation $x$ on $N$ is in the \emph{strict core (SC)} of market $M$ if it admits no weakly blocking coalition. An allocation that is in the strict core is also said to be \emph{strict core stable}.
    An allocation is \emph{Pareto optimal (PO)} if $N$ is not a weakly blocking coalition.
     It is clear that strict core implies core and also Pareto optimality. Core implies weak Pareto optimality and also individual rationality.

    \emph{From now on we will assume 1-0 utilities for all the statements. }

    \begin{example}
        Consider a house allocation with existing tenants setting:
        \begin{itemize}
            \item $N=\{1,2,3,4,5\}$;
            \item $H=\{h_1,h_2,h_3,h_4, h_5, h_6\}$;
            \item $\w(i)=h_i$ for $i\in \{1,\ldots, 4\}$;
            \item $A_1=\{h_2\}$; $A_2=\{h_3\}$; $A_3=\{h_1\}$, $A_4=\{h_5\}$, $A_5=\{h_5, h_6\}$  
        \end{itemize}
        Then a feasible SIR outcome and Pareto optimal outcome is allocation $x$ such that $x(1)=h_2$, $x(2)=h_3$, $x(3)=h_1$, $x(4)=h_4$, and  $x(5)=h_5$.
        
        Then a feasible IR outcome and Pareto optimal outcome is allocation $y$ such that $y(1)=h_2$, $y(2)=h_3$, $y(3)=h_1$, $y(4)=h_5$, and  $y(5)=h_6$.
          
        \end{example}
    
    It is obvious that  S-IR implies IR.
    \begin{proposition}
        S-IR implies IR.
        \end{proposition}
    
    \begin{proposition}
       An allocation that maximizes welfare subject to S-IR may have less welfare than an allocation that maximizes welfare subject to IR. 
        \end{proposition}
        \begin{proof}
           Consider the setting in which there are two agents and both agents have zero value houses but one agent wants the  other agent's house. Then the only feasible S-IR allocation is $\w$ whereas there exists an IR allocation in which one agents gets an acceptable house with utility 1. 
            \end{proof}

    \section{Mechanisms}
    
    We present two mechanisms: MSIR and MIR. MSIR satisfies the maximum number of agents subject to the SIR constraint. MIR satisfies the maximum number of agents subject to the IR constraint. Both algorithms are based on constructing a bipartite graph that admits a perfect matching and repeatedly modifying the graph while still ensuring the the graph still has a maximum weight perfect matching. 
    
    \subsection{The MSIR Mechanism}
    
    The MSIR mechanism is specified as Algorithm~\ref{algo:msir}.
    
    % \begin{example}
    %
    %     \end{example}
    %

    		\begin{algorithm}[h!]
    	        % \footnotesize
    		  \caption{MSIR}
    		  \label{RAR-algo}
    		\renewcommand{\algorithmicrequire}{\wordbox[l]{\textbf{Input}:}{\textbf{Output}:}} 
    		 \renewcommand{\algorithmicensure}{\wordbox[l]{\textbf{Output}:}{\textbf{Output}:}}
    		\begin{algorithmic}
    		%	\small
    			\REQUIRE $(N,H,\succsim ,\w)$
    			\ENSURE 
    		\end{algorithmic}
    		\algsetup{linenodelimiter=\,}
    		  \begin{algorithmic}[1] 
                  \STATE $G=(A\cup B, E,w)$.
                  Set $k$ to $|n-m|$. If $m>n$, then $A=N\cup D_N$ where $D_N=\{d_1,\ldots, d_k\}$. If $n>m$, then $B=H\cup D_H$ where $D_H=\{o_1,\ldots, o_k\}$. $E$ is defined as follows:
                  \begin{itemize}
                      \item For $a\in N$ and $b=\w(a)$, add $\{a,b\}\in E$  with $w(\{a,b\})=0$ if $b\notin A_a$ and $w(\{a,b\})=1$ if $b\in A_a$.
                     \item For $a\in N$ such that $\w(a)\neq null$, $\w(a)\notin A_a$, and $b=H\setminus \{\w(a)\}$ such that $b\in A_a$, add $\{a,b\}\in E$ with $w(\{a,b\})=1$.
    %                  \item For $a\in N$ such that $\w(a)\neq null$ and $b=H\setminus \{\w(a)\}$ such that $b\in A_a$, add $\{a,b\}\in E$ with $w(\{a,b\})=1$.
                        \item For $a\in N$ such that $\w(a)=null$ and $b\in B$, add $\{a,b\}\in E$ with %$w(\{a,b\})=0$.
$w(\{a,b\})=0$ if $b\notin A_a$ and $w(\{a,b\})=1$ if $b\in A_a$.
                        \item For $b\in D_H$ and  $a\in \{a\in N\midd \w(a)=null\}$, add  $\{a,b\}\in E$ 
with $w(\{a,b\})=0$ if $b\notin A_i$ and $w(\{a,b\})=1$ if $b\in A_a$.                        
                        %with $w(\{a,b\})=0$.
                                       \item For $a\in D_N$ and each $b\in B$, add  $\{a,b\}\in E$ with $w(\{a,b\})=0$.
                  \end{itemize}
              \COMMENT{Note that an agent with an acceptable endowment only has an edge to his endowment. An agent with an unacceptable endowment only has edges with his endowment as well as the acceptable houses.}
                  % \begin{itemize}
    %                   \item $\{a,b\}\in E$ if $a\in N$, $b\in H$, and $u_a(b)=1$. In that case $w(\{a,b\})=1$.
    %               \item $\{a,b\}\in E$ if $a\in N$, $b\in H$, and $b= \w(a)$. In that case $w(\{a,b\})=u_a(b)$.
    %               \item $\{a,b\}$ for $a\in D_N$ and $b\in B$. In that case $w(\{a,b\})=0$.
    %                                     \item $\{a,b\}$ for $b\in D_H$ and $a\in A$. In that case $w(\{a,b\})=0$.
    %               \end{itemize}
                  \STATE Compute the maximum weight perfect matching of $G$. Let the weight of the matching be $W$.
                  \STATE Take permutation $\pi$ of $N$
                  \FOR{$i$=1 to $n$}
                  \STATE Set $t_i$ to 1.
                  \STATE Remove from $G$ each edge $\{\pi(i),b\}\in E$  such that
                   $w(\{\pi(i),b)\})=0$. Compute the maximum weight perfect matching of $G$. If the weight is less than $W$ or if there does not exist a perfect matching, then  put back the removed edges and set $t_i$ to 0.              
                  % \IF{$\w(\pi(i)\neq null$}
      %             \STATE Remove edge $\{\pi(i),\w(\pi(i))\}$ from $G$ if $u_{\pi(i)}(\w(\pi(i)))=0$. Compute the maximum weight perfect matching of $G$. If the weight is less than $W$ or if there does not exist a perfect matching, then put back the edge $\{\pi(i),\w(\pi(i))\}$.
      %             \ELSE
      %            \STATE Remove each edge $\{\pi(i),b)\}$ from $G$ if
      %             $w(\{\pi(i),b)\})=0$. Compute the maximum weight perfect matching of $G$. If the weight is less than $W$ or if there does not exist a perfect matching, then put back the edge $\{\pi(i),\w(\pi(i))\}$.
      %             \ENDIF
                  \ENDFOR
                 \STATE Compute the maximum weight perfect matching $M$ of $G$. Consider the allocation $x$ in which each agent $i\in N$ gets a house that it is matched to in $M$. If $i\in N$ is matched to a dummy house, its allocation is null.  
      \RETURN $(x(1),\ldots, x(n))$
    		 \end{algorithmic}
    		\label{algo:msir}
    		\end{algorithm}

            \begin{proposition}
                MSIR runs in polynomial-time.% $O({(n+m)}^3n)$.
                \end{proposition}
                \begin{proof}
                    The graph $G$ has $O(n+m)$ vertices. The graph is modified at most $O(n)$ times. Each time, a maximum weight perfect matching is computed that takes polynomial time ${|V(G)|}^3$~\citep{KoVy06a}. 
                    \end{proof}

                \begin{proposition}
                    MSIR returns an allocation that is SIR.
                    \end{proposition}
                    \begin{proof}
                        Throughout the algorithm, we make sure that $G$ admits a perfect matching. If a modification to $G$ leads to a lack of a perfect matching, then such a modification is reversed. 
An agent with an endowment only has an edge to his endowment or to an acceptable house. Therefore while a perfect matching exists, an agent with an unacceptable endowment cannot be matched to an unacceptable house other than his endowment. Thus MSR  returns an allocation that is SIR.
                        \end{proof}
                    
                  \begin{proposition}
                      MSIR returns an allocation that satisfies the maximum number of agents subject to SIR.
                      \end{proposition}
                      \begin{proof}
                           Throughout the algorithm, we make sure that $G$ admits a perfect matching which ensures that the corresponding allocation satisfies S-IR. An agent with an acceptable endowment does not have an edge to any other house so he has to be allocated his endowment. An agent with an unacceptable endowment either has an edge to his endowment or to houses that are acceptable. Hence, in a perfect matching, the agent either gets his endowment or an acceptable house.
 Given this condition, we compute the maximum weight perfect matching. This implies that the corresponding allocation satisfies the maximum number of agents under the SIR constraint. 
                          \end{proof}

                      \begin{corollary}
                          MSIR returns an allocation that is Pareto optimal among the set of SIR allocations.
                          \end{corollary}
                          \begin{proof}
                              Assume for contradiction that the MSIR allocation is Pareto dominated by an SIR allocation. But this means that the MSIR allocation does not satisfy the maximum number of agents. 
                              \end{proof}

                \begin{proposition}
                    MSIR is strategyproof.
                    \end{proposition}
                                   \begin{proof}
MSIR returns a perfect matching of weight $W$. During the running of MSIR, each time a modification is made to the graph $G$, it is is ensured that $G$ admits a perfect matching of weight $W$. 
               Assume for contradiction that MSIR is not strategyproof and some agent $i\in N$ with turn $k$ in permutation $\pi$ gets a more preferred house when he misreports. This means that agent $i$ gets an unacceptable house when he reports the truthful preference $A_i$. Let the allocation be $x$.  
This implies that in permutation $\pi$, when $i$'s turn comes, there exists no feasible maximum weight perfect matching of size $W$ in which $i$ gets an acceptable house and each agent preceding $i$ in permutation $\pi$ gets $t_i$ acceptable houses. Since $i$ can get a more preferred house by misreporting, $i$ gets an acceptable house if he reports $A_i'$. Let such an allocation be $x'$. Note that $x'$ is a feasible maximum weight perfect matching even when $i$ tells the truth and even if each agent preceding $i$ in permutation $\pi$ gets $t_i$ acceptable houses. But this is a contradiction because there does exist a feasible maximum weight perfect matching of size $W$ in which $i$ gets an acceptable house and each agent preceding $i$ in permutation $\pi$ gets $t_i$ acceptable houses.
                                       \end{proof}

                  \begin{lemma}
                      For 1-0 utilities, any S-IR welfare maximizing allocation is core stable. 
                      \end{lemma}
                      \begin{proof}
  Assume that there is a blocking coalition $S$. It can only consist of agents who did not get an acceptable house in the allocation. S-IR implies that agents who originally own an acceptable house keep the acceptable house.  If an agent $i\in N$ is in $S$ who does not originally own any house, he cannot be part of $S$ because he nas nothing to give to other agents, so other agents can satisfy each other without letting $i$ be a member of $S$. Therefore, $S$ consists of those agents who owned an unacceptable house and are allocated an unacceptable house. Due to S-IR, agents in $S$ are allocated their own house. Now if $S$ admits a blocking coalition this implies that the S-IR welfare maximizing allocation was not  S-IR welfare maximizing which is a contradiction.                        
                          \end{proof}
              
                          \begin{proposition}
                              MIR returns an allocation that is core stable. 
                              \end{proposition}
 %  \begin{proposition}
 % MSIR is core stable for 1-0 utilities.
 %      \end{proposition}
 %
 \subsection{The MIR Mechanism}
 
     The MSIR mechanism is specified as Algorithm~\ref{algo:mir}.
 
     		\begin{algorithm}[h!]
     	        % \footnotesize
     		  \caption{MIR}
     		  \label{RAR-algo}
     		\renewcommand{\algorithmicrequire}{\wordbox[l]{\textbf{Input}:}{\textbf{Output}:}} 
     		 \renewcommand{\algorithmicensure}{\wordbox[l]{\textbf{Output}:}{\textbf{Output}:}}
     		\begin{algorithmic}
     		%	\small
     			\REQUIRE $(N,H,\succsim ,\w)$
     			\ENSURE 
     		\end{algorithmic}
     		\algsetup{linenodelimiter=\,}
     		  \begin{algorithmic}[1] 
                   \STATE $G=(A\cup B, E,w)$.
              Set $k$ to $|n-m|$. If $m>n$, then $A=N\cup D_N$ where $D_N=\{d_1,\ldots, d_k\}$. If $n>m$, then $B=H\cup D_H$ where $D_H=\{o_1,\ldots, o_k\}$. $E$ is defined as follows:
                   \begin{itemize}
                       \item For $a\in N$ and $b=\w(a)$, add $\{a,b\}\in E$ with $w(\{a,b\})=0$ if $b\notin A_a$ and $w(\{a,b\})=1$ if $b\in A_a$.
                      \item For $a\in N$ and $b=H\setminus \{\w(a)\}$ such that $b\in A_a$, add $\{a,b\}\in E$ with $w(\{a,b\})=1$.
                         \item For $a\in N$ such that $\w(a)=null$ and $b\in B$, add $\{a,b\}\in E$ 
with $w(\{a,b\})=0$ if $b\notin A_a$ and $w(\{a,b\})=1$ if $b\in A_a$.                        
                       %   $w(\{a,b\})=0$.
                         \item For $a\in N$ such that $\w(a)=h$ and $h\notin A_a$ and each $b\in B$, add $\{a,b\}\in E$ with
with $w(\{a,b\})=0$ if $b\notin A_a$ and $w(\{a,b\})=1$ if $b\in A_a$.                        
                          %$w(\{a,b\})=0$.
                         %  \item For $b\in D_H$ and each $a\in N$, add  $\{a,b\}\in E$ with $w(\{a,b\})=0$.
                         \item For $b\in D_H$ and each $a\in \{a\in N\midd \w(a)\notin A_a\}$, add  $\{a,b\}\in E$ with $w(\{a,b\})=0$.
                                        \item For $a\in D_N$ and each $b\in B$, add  $\{a,b\}\in E$ with $w(\{a,b\})=0$.
                   \end{itemize}
               \COMMENT{Note that an agent with an acceptable endowment only has an edge to his endowment or to an acceptable house. An agent with an unacceptable endowment has an edge to every $b\in B$.}
                   % \begin{itemize}
     %                   \item $\{a,b\}\in E$ if $a\in N$, $b\in H$, and $u_a(b)=1$. In that case $w(\{a,b\})=1$.
     %               \item $\{a,b\}\in E$ if $a\in N$, $b\in H$, and $b= \w(a)$. In that case $w(\{a,b\})=u_a(b)$.
     %               \item $\{a,b\}$ for $a\in D_N$ and $b\in B$. In that case $w(\{a,b\})=0$.
     %                                     \item $\{a,b\}$ for $b\in D_H$ and $a\in A$. In that case $w(\{a,b\})=0$.
     %               \end{itemize}
                   \STATE Compute the maximum weight perfect matching of $G$. Let the weight of the matching be $W$.
                   \STATE Take permutation $\pi$ of $N$
                   \FOR{$i$=1 to $n$}
                       \STATE Set $t_i$ to 1.
                   \STATE Remove from $G$ each edge $\{\pi(i),b\}\in E$ such that
                    $w(\{\pi(i),b)\})=0$. Compute the maximum weight perfect matching of $G$. If the weight is less than $W$ or if there does not exist a perfect matching, then put  back the removed edges and set $t_i$ to 0.                        
                   % \IF{$\w(\pi(i)\neq null$}
       %             \STATE Remove edge $\{\pi(i),\w(\pi(i))\}$ from $G$ if $u_{\pi(i)}(\w(\pi(i)))=0$. Compute the maximum weight perfect matching of $G$. If the weight is less than $W$ or if there does not exist a perfect matching, then put back the edge $\{\pi(i),\w(\pi(i))\}$.
       %             \ELSE
       %            \STATE Remove each edge $\{\pi(i),b)\}$ from $G$ if
       %             $w(\{\pi(i),b)\})=0$. Compute the maximum weight perfect matching of $G$. If the weight is less than $W$ or if there does not exist a perfect matching, then put back the edge $\{\pi(i),\w(\pi(i))\}$.
       %             \ENDIF
                   \ENDFOR
                  \STATE Compute the maximum weight perfect matching $M$ of $G$. Consider the allocation $x$ in which each agent $i\in N$ gets a house that it is matched to in $M$. If $i\in N$ is matched to a dummy house, its allocation is null.  
       \RETURN $(x(1),\ldots, x(n))$
     		 \end{algorithmic}
     		\label{algo:mir}
     		\end{algorithm}

            \begin{proposition}
                MIR runs in polynomial-time. % $O({(n+m)}^3n)$.
                \end{proposition}
                \begin{proof}
                    The graph $G$ has $O(n+m)$ vertices. The graph is modified at most $O(n)$ times. Each time, a maximum weight perfect matching is computed that takes time ${|V(G)|}^3$~\citep{KoVy06a}. 
                    \end{proof}
     
     \begin{proposition}
         MIR returns an allocation that is IR.
         \end{proposition}
                             \begin{proof}
                                 Throughout the algorithm, we make sure that $G$ admits a perfect matching. If a modification to $G$ leads to a lack of a perfect matching, then such a modification is reversed. 
         An agent with an acceptable endowment only has an edge to his endowment or to other acceptable houses. Therefore while a perfect matching exists, an agent with an acceptable endowment can only be matched to an acceptable house. Thus MSR  returns an allocation that is IR.
                                 \end{proof}
         
       \begin{proposition}
           MIR returns an allocation that satisfies the maximum number of agents subject to IR.
           \end{proposition}
           \begin{proof}
                Throughout the algorithm, we make sure that $G$ admits a perfect matching which ensures that the corresponding allocation satisfies IR. Given this condition, we compute the maximum weight perfect matching. This implies that the corresponding allocation satisfies the maximum number of agents under the IR constraint.

               \end{proof}

               \begin{lemma}
                  An allocation that maximizes welfare subject subject to IR has the same welfare even if IR is not imposed. 
                   \end{lemma}
                   \begin{proof}
                       Let the set of agents who own an acceptable house be $A$ and $B$ is $N\setminus A$. Then an allocation with maximum welfare subject to IR is one that satisfies each agent in $A$ and maximum number of agents $c$ in $B$. Let us assume that without requiring IR, the welfare is more than $|A|+c$. Then this means that there is an allocation that does not satisfy $k$ agents in $A$ but satisfies $c+k+1$ agents in $B$. But this means that by not satisfying $k$ agents in $A$, $k+1$ agents in $B$ can be satisfied. But this is a contradiction because the sacrificing $k$ agents in $A$ can forego exactly $k$ houses for agents in $B$. 
                                          \end{proof}

           \begin{proposition}
               MIR returns an allocation that satisfies the maximum number of agents.
               \end{proposition}
           % \begin{corollary}
%                MSIR returns an allocation that is Pareto optimal among the set of IR allocations.
%                \end{corollary}
               \begin{corollary}
                   MIR returns an allocation that is Pareto optimal.
                   \end{corollary}
               \begin{proof}
                   MIR returns an allocation that is Pareto optimal among the set of IR allocations. Now if the allocation were not Pareto optimal, any allocation that Pareto dominated it would be IR as well. But this is a contradiction.  
                   \end{proof}

                \begin{proposition}
                    MIR is strategyproof.
                    \end{proposition}
                                                       \begin{proof}
                    MIR returns a perfect matching of weight $W$. During the running of MIR, each time a modification is made to the graph $G$, it is is ensured that $G$ admits a perfect matching of weight $W$. 
                                   Assume for contradiction that MSIR is not strategyproof and some agent $i\in N$ with turn $k$ in permutation $\pi$ gets a more preferred house when he misreports. This means that agent $i$ gets an unacceptable house when he reports the truthful preference $A_i$. Let the allocation be $x$.  
                    This implies that in permutation $\pi$, when $i$'s turn comes, there exists no feasible maximum weight perfect matching of size $W$ in which $i$ gets an acceptable house and each agent preceding $i$ in permutation $\pi$ gets $t_i$ acceptable houses. Since $i$ can get a more preferred house by misreporting, $i$ gets an acceptable house if he reports $A_i'$. Let such an allocation be $x'$. Note that $x'$ is a feasible maximum weight perfect matching even when $i$ tells the truth and even if each agent preceding $i$ in permutation $\pi$ gets $t_i$ acceptable houses. But this is a contradiction because there does exist a feasible maximum weight perfect matching of size $W$ in which $i$ gets an acceptable house and each agent preceding $i$ in permutation $\pi$ gets $t_i$ acceptable houses.
                                                           \end{proof}

           % Does that imply PO?
           %
           %
         %   \section{Only Endowments}

        %We consider the setting in which 
    
              % \begin{proposition}
              %     Core stability implies SIR.
              %     \end{proposition}
              %
              %
              %     DOES IT???

                 \begin{proposition}
                    For 1-0 utilities, an IR welfare maximizing allocation may not be core stable. 
                 \end{proposition}
                    \begin{proof}
                        Consider a four agent setting in which each agent owns a house unacceptable to himself. Agent $2$ owns a house that is acceptable to $4$ and $1$ and agent $1$ owns a house that is acceptable to $2$ and $3$. Consider an allocation in which agent $3$ gets $1$'s house and $4$ gets $2$'s house. Such an allocation is IR and satisfies the maximum number of agents. 
However, it is not core stable because $1$ and $2$ can form a blocking coalition. 
      \end{proof}

        \section{Conclusions}
        
        In this paper, two new mechanisms called MSIR and MIR were introduced. % that have relative merits over previous mechanisms. 
        See Table~\ref{table:summary} for a summary of properties satisfied by different mechanisms for house allocation with existing tenants. 
       
       \begin{table}[h!]
       %    \small
           \centering
               \scalebox{1}{
       \centering
               \begin{tabular}{lcccc}
               \toprule
                       &\textbf{MSIR}&\textbf{MIR}\\
               \midrule
               S-IR&+&-\\
                 IR&+&+\\
                   Core&+&-\\
                       \midrule
               Pareto optimal&-&+\\
               Max welfare subject to SIR&+&-\\
                          Max welfare subject to IR&-&+\\
                Max welfare&-&+\\
               \bottomrule
               \end{tabular}
               }
               \caption{Properties satisfied by mechanisms for house allocation with existing tenants with dichotomous preferences. 
               }
               \label{table:summary}
               \end{table}

        % \begin{table}[h!]
 %        %    \small
 %            \centering
 %                \scalebox{1}{
 %        \centering
 %                \begin{tabular}{lcccc}
 %                \toprule
 %                        &\textbf{MSIR}&\textbf{MIR}&TCR&Max Matching\\
 %                \midrule
 %                S-IR&+&-&-&-\\
 %                  IR&+&+&+&-\\
 %                    Core&+&-&+&-\\
 %                        \midrule
 %                Pareto optimal&-&+&+&+\\
 %                Max welfare subject to SIR&+&-&-&-\\
 %                           Max welfare subject to IR&-&+&-&-\\
 %                 Max welfare&-&+&-&+\\
 %                \bottomrule
 %                \end{tabular}
 %                }
 %                \caption{Properties satisfied by mechanisms for house allocation with existing tenants with dichotomous preferences.
 %                }
 %                \label{table:summary}
 %                \end{table}

      \section*{Acknowledgments}

     The author is funded by the Australian Government through the Department of Communications and the Australian Research Council through the ICT Centre of Excellence Program. The authors thanks P{\'{e}}ter Bir{\'{o} and Greg Plaxton for useful comments and pointers. 
                
        % \bibliography{../../pamas/abb,../../pamas/group,../../pamas/brandt,../../pamas/aziz}

		\end{document}